\def\CommentVersion{}

\documentclass{article}
\usepackage{booktabs} 

\usepackage{amsthm}
\usepackage{graphicx}
\usepackage{tikz}
\usetikzlibrary{plotmarks,arrows,automata,positioning,fit,shapes.geometric,backgrounds}

\usepackage{pgf}

\usepackage{pgfplots}
\usepackage{pgfkeys}
\pgfplotsset{compat=1.13}

\usepackage{algpseudocode}
\usepackage{algorithm}
\usepackage{caption}
\usepackage{multirow}
\usepackage{subcaption}
\usetikzlibrary{calc}
\usepackage{relsize}
\usepackage{url}
\tikzset{fontscale/.style = {font=\relsize{#1}}}    
\usepackage{adjustbox}
\usepackage{capt-of}
\usepackage{mathtools}
\usepackage{amsmath,amsfonts,amssymb}
\usepackage{proof}

\ifdefined\CommentVersion
  \usepackage{marginnote}
  \usepackage{color}
  \usepackage{soul}
  \newcommand{\marginX}{\marginnote{\huge{\quad\quad\textbf{!}\quad\quad}}}
  \newcommand{\he}[1]{{\color{green}\marginX{}[\textbf{Houssam}: #1]}}
  \newcommand{\gl}[1]{{\color{magenta}\marginX{}[\textbf{Giuseppe}: #1]}}
  \newcommand{\ro}[1]{{\color{blue}\marginX{}[\textbf{Richard}: #1]}}

  \setstcolor{blue}

  \newcommand{\deleted}[1]{\st{#1}}
\else
  \newcommand{\he}[1]{}
  \newcommand{\gl}[1]{}
  \newcommand{\ro}[1]{}
  
  \newcommand{\deleted}[1]{}
  
\fi

\ifdefined\ReportVersion
  \newcommand{\paperOnly}[1]{}
  \newcommand{\reportOnly}[1]{#1}
\else
  \newcommand{\paperOnly}[1]{#1}
  \newcommand{\reportOnly}[1]{}
\fi

\newcommand{\taskset}[1]{\mathcal{T}_{#1}}
\newcommand{\task}[1]{\tau_{#1}}

\newcommand{\deadline}[1]{\mathsf{D}_{#1}}
\newcommand{\period}[1]{\mathsf{T}_{#1}}

\newcommand{\job}[2]{\mathcal{J}_{#1,#2}}
\newcommand{\arrival}[2]{\mathsf{a}_{#1,#2}}
\newcommand{\server}[1]{\mathcal{S}_{#1}}
\newcommand{\serveru}[1]{\mathsf{u}_{#1}}
\newcommand{\budget}[1]{\mathsf{Q}_{#1}}
\newcommand{\serverp}[1]{\mathsf{P}_{#1}}
\newcommand{\serverd}[1]{\mathsf{d}_{#1}}
\newcommand{\charge}[2]{\mathsf{C}_{#1,#2}}
\newcommand{\finishtime}[2]{\mathsf{f}_{#1,#2}}

\newcommand{\virtualtime}[1]{\mathsf{V}_{#1}}
\newcommand{\bandwith}[1]{\mathsf{U}_{#1}}

\newcommand{\grubstate}[1]{{\textsc{#1}}}

\makeatletter
\tikzset{
    scale plot marks/.is choice, scale plot
    marks/false/.code={ \def\pgfuseplotmark##1{\pgftransformresetnontranslations\csname
    pgf@plot@mark@##1\endcsname} }, scale plot marks/true/.style={},
    scale plot marks/.default=true }
\makeatother

\newtheorem{mytheo}{\textbf{Theorem}}
\newtheorem{mylemma}{\textbf{Lemma}}
\newtheorem{mycollo}{\textbf{Corollary}}

\title{Migrate when necessary: toward partitioned reclaiming for soft real-time tasks}
\author{Houssam-Eddine Zahaf, Giuseppe Lipari \& Luca Abeni}
\date{October 7, 2017}

\begin{document}

\maketitle

\begin{abstract}
  This paper presents a new strategy for scheduling soft real-time
  tasks on multiple identical cores. The proposed approach is based on
  partitioned CPU reservations and it uses a reclaiming mechanism to
  reduce the number of missed deadlines. We introduce the possibility
  for a task to temporarily migrate to another, less charged, CPU when
  it has exhausted the reserved bandwidth on its allocated CPU. In
  addition, we propose a simple load balancing method to decrease the
  number of deadlines missed by the tasks. The proposed algorithm has been evaluated through
  simulations, showing its effectiveness (compared to other multi-core
  reclaiming approaches) and comparing the performance of different
  partitioning heuristics (Best Fit, Worst Fit and First Fit).
\end{abstract}

\textbf{Keywords: }{Resource reservations, partitioned scheduling, CPU
  reclaiming
  
\section{Introduction}
\label{sec:introduction}

Many mobile appliances (smart-phones, tablets, smart-watches, etc.)
feature variable workloads which may contain (soft) real-time tasks
(audio, video, real-time communication, gaming, etc); they have strong
requirements in terms of energy consumption, but also of latency,
memory etc. They are equipped with multicore processors, and run with
a general purpose OS (e.g. Android, which is based on the Linux
kernel, or iOS). 

These OSs generally provide some kind of real-time support for
time sensitive applications. For example, the Linux kernel provides
three real-time scheduling policies: \texttt{SCHED\_FIFO} and
\texttt{SCHED\_RR} which are based on fixed priorities and are
standardized by POSIX; and \texttt{SCHED\_DEADLINE}~\cite{Lelli16},
which provides resource reservation on top of Earliest Deadline First
(EDF). Thanks to the \emph{temporal isolation} property, the latter is
particularly useful when mixing real-time and non real-time workloads
in the same system, as required by the appliances mentioned above.
From our experience, the typical workload executed in these systems
consists of a few hard real-time tasks with relatively small
utilizations, and a certain number of soft real-time tasks with
variable execution time.

These scheduling policies are highly configurable, and can be used to
implement either \emph{global scheduling} or \emph{partitioned
  scheduling} by properly setting the tasks' \emph{CPU affinities} or
using the {\tt cpuset} mechanism\footnote{In theory, CPU affinities
  can even be used to implement even more complex configurations, but
  this setup is harder to analyze~\cite{Gujarati2015}.}.

From a logical point of view, in global scheduling all tasks are
ordered by priority in a single global queue, and the $m$ highest
priority tasks are executed on the $m$ available processors.  Due to
	the intrinsic nature of the global queue, tasks may \emph{migrate}
from one processor to another one. Task migration enables automatic
load balancing across CPUs, however it is a possible source of
overhead: in addition to the cost of the context switch, tasks may
have to reload the content of at least the L1 cache; in turns this may
cause additional contention on the shared memory bus, cache evictions
for the other tasks, increased execution times, etc.  Moreover, global
EDF and global FP algorithms do not always achieve a high schedulable
utilization.
 
In partitioned scheduling, every task is allocated on a processor and
it is not allowed to migrate. 
Partitioned scheduling may achieve a higher schedulable
utilization than global EDF and global fixed priorities, and it is
simpler to implement in an OS kernel. However, pure (and static)
partitioned scheduling performs poorly in presence of highly
dynamically workloads, where tasks have variable execution times and
can dynamically enter and leave the system. In some extreme
cases, a re-allocation of tasks to processors (\emph{load balancing})
might be needed every time the system load changes (a task enters (or
leaves) the system).

Recently, it has been shown that semi-partitioned scheduling can
achieve in practice very high (quasi-optimal) schedulable
utilization~\cite{BBB16-semipartitioning}, even in presence of dynamic
workloads~\cite{Alessandro-ecrts}. Using semi-partitioned scheduling,
tasks are split in two or more sub-tasks, and every sub-task is
statically assigned to a CPU / core. This solution has been shown to
work well in practice; however, to enforce precedence constraints
across sub-tasks of the same task, the task algorithm assigns short
relative deadlines to the sub-tasks and imposes that a sub-task cannot
start executing before the absolute deadline of the preceding
one. This complicates the scheduling algorithm and the
schedulability analysis (and hence admission control of newly incoming
tasks).

When considering Open Systems running soft real-time tasks, it may be
more convenient to avoid the complexity of task splitting altogether,
and use other simpler techniques, such as resource reservation and
reclaiming.
We focus on this proposal because our goal is to support mixed workloads
(composed by both hard and soft real-time tasks): resource reservation
can be used to schedule both hard and soft real-time tasks~\cite{abeni1998integrating}
and a reclaiming mechanism can improve the performance of soft real-time
tasks.


\paragraph{Resource Reservation and reclaiming}
Resource reservations are used to provide temporal isolation and
individual real-time guarantees to hard and soft real-time tasks. In
the resource reservation framework \cite{Mer94}, every task is
assigned a reservation period and it is reserved a percentage of the
computational bandwidth: if its resource requirements never exceed the
reserved bandwidth, and by properly assigning the reservation period,
the task is guaranteed to never miss its deadlines. 

On the other hand, we can also reserve bandwidth based on the average
computational requirements of the (soft) real-time tasks. In this
case, a task may occasionally miss its deadline. However, we can further
opportunistically reduce the percentage of deadline misses by
reclaiming the unused bandwidth in the system.

There is a vast literature on reclaiming for resource reservation
systems in single processor systems. The problem has been treated much
less extensively in multiprocessor systems. Currently, Linux supports
resource reservations with the SCHED\_DEADLINE policy (both global and
partitioned) and implements the sequential M-GRUB reclaiming
strategy~\cite{abeni2016multicore} that, as will be shown in
Section~\ref{sec:comp-job-migr}, can result in a large number of
tasks' migrations.

\paragraph{Our approach} In this paper, we investigate the use of
\emph{partitioned resource reservations} together with a CPU
reclaiming mechanism and a limited (temporary) task migration, plus a
simple load balancing algorithm. The simplest way to combine CPU
reclaiming with partitioned resource reservations is to statically
assign tasks/reservations to CPUs, and to perform per-CPU (local)
reclaiming only. This is basically what we get by using the single
processor GRUB algorithm~\cite{lipari2000greedy}. Such a simple
strategy is ineffective when a soft real-time task needing more than
what has been reserved is placed on a CPU where there is no unused
bandwidth, whereas unused bandwidth is available on other CPUs.

On the other end of the spectrum, we can use resource reservations
with global EDF scheduling and an extension of the GRUB algorithm for
global scheduling. This solution has been investigated
in~\cite{abeni2016multicore}: two reclaiming algorithms have been proposed, the
\emph{sequential reclaiming} (where unused bandwidth is collected on a
per-CPU basis) and \emph{parallel reclaiming} (where unused bandwidth
is collected at the global level). The latter is very similar to the
M-CASH algorithm proposed by Pellizzoni and Caccamo~\cite{pellizzoni2008m}.

The solution investigated in this paper is a trade-off between these
two extremes. We use partitioned reservations (tasks / reservations
statically assigned to CPUs / cores) and global reclaiming (achieved
through a temporary migration mechanism). We complement the temporary
migration mechanism with a load balancing algorithm to better tolerate
large variations in the workload.

\paragraph{Organization}
The paper is structured as follows. In Section \ref{sec:state-art}, we
discuss the state of the art in multiprocessor real-time scheduling.
In Section \ref{sec:sys-model} we present the system model. In Section
\ref{sec:grub-algorithm} we recall the Greedy Reclamation of Unused
Bandwidth (GRUB) algorithm. Section~\ref{sec:grub-tempmigr}
presents our proposed framework and it is the core of our paper: we
describe the migration conditions, and we
present the proof that the migration will not affect the correctness
of temporal isolation.
Section~\ref{sec:grub-partitioning} describes the partitioning
and load balancing techniques used when temporary migrations are not
sufficient to provide good real-time performance.
Section \ref{sec:results-discussions} is
reserved to the experimentation and discussions. We conclude in
Section \ref{sec:conclusion}.

\section{State of the art}
\label{sec:state-art}

\paragraph{Optimality} 
In real-time scheduling theory, all optimal
scheduling algorithms for multiprocessor systems (for example,
RUN~\cite{regnier2011run}, QPS~\cite{Massa2014}, etc.) belong to the class
of global scheduling algorithms. They are optimal in the sense that,
under certain conditions, they can achieve full utilization of the
platform without any deadline miss. From a practical
point of view these algorithms result to be complex to implement, because they
require tight synchronization among processors. Furthermore, they all
produce \emph{frequent task migrations} which may increase tasks'
execution times.

\paragraph{Partitioning} 
On the other end of the spectrum we have partitioned
scheduling. Partitioned scheduling is simpler to implement as each
processor can be considered as a single independent computing
resource. Also, if each task is allocated on one processor, it cannot
migrate, thus reducing the scheduling overhead. Unfortunately, optimal
partitioning is a NP-hard problem (equivalent to knapsack).


Recently, semi-partitioned scheduling has emerged as an interesting
alternative to global and partitioned scheduling. In semi-partitioned
scheduling, most of the tasks are partitioned among processors, and a
few tasks are \emph{split} in multiple parts (with each part executing on
on a different processor). In comparison
with global scheduling, semi-partitioned scheduling permits to reduce
and keep under control the number of migrations; semi-partitioned
scheduling also shows a better utilization factor when compared with
pure partitioned scheduling. Brandenburg and Gül
\cite{BBB16-semipartitioning} have shown that the use of
semi-partitioned scheduling coupled with a slack reclaiming
strategy~\cite{Caccamo05} in practice allows to achieve very high
utilization factors.

Task splitting is usually performed on a static task set before
run-time, thus this technique cannot be easily used on-line on a
dynamically varying task set. Recent
work~\cite{BBB16-semipartitioning,Alessandro-ecrts} addresses this
issues by performing task splitting on dynamically arriving tasks.

\paragraph{Soft Real-Time scheduling} Since respecting all of the
tasks' deadlines on multi-processor systems is a complex problem,
some scheduling algorithms focus on providing tardiness / lateness
guarantees to soft real-time tasks. Many of these algorithms are based
on global scheduling, that has been shown to have some optimality properties
for soft real-time tasks~\cite{Devi05,Val05}.

Scheduling of soft and hard real-time tasks has been mixed by using
partitioned scheduling for hard tasks and global scheduling for soft
tasks (and scheduling soft tasks in background respect to hard
tasks)~\cite{Bra07}. In this work, a reclaiming mechanism similar to CASH
has been used to improve the performance of soft real-time tasks.

Finally, it has been noticed that semi-partitioned scheduling is a good
choice for soft real-time tasks too~\cite{And14}.

As explained, in this paper we focus on a different approach, based on
partitioned scheduling (with some possible load balancing) and CPU reclaiming
associated to temporary migrations to reduce the number of migrations.

\paragraph{Resource reclaiming} Many reclaiming algorithms exists for
reclaiming unused bandwidth on single processors
\cite{lin2005improving,Caccamo05,lipari2000greedy}. In particular, we
cite here two reclaiming mechanisms based on EDF. The CASH (Capacity
Sharing) \cite{Caccamo05} algorithm utilizes a queue of unused
capacities, each capacity is a pair of budget and deadline. When a job
of a periodic task finishes, the residual budget together with the
task absolute deadline are inserted in the reclaiming queue. When a
task executes, in addition to its own budget, it can use all budgets
in the reclaiming queue with deadline no greater than the task's
deadline. The algorithm is simple and effective, however it can only
reclaim the unused budget released by periodic tasks.

GRUB (Greedy Reclamation of Unused Bandwidth) \cite{lipari2000greedy}
uses a different scheme based on utilization. The algorithm keeps
track of the bandwidth of the active tasks in the system: when a task
executes, the budget is decreased by an amount proportional to the
free available bandwidth. GRUB can be effectively used with periodic
and sporadic tasks, however it can only treat periodic reservations
with relative deadline equal to period. The techniques based on this
paper are based on GRUB (see Section \ref{sec:grub-algorithm}.)

M-CASH is a reclaiming mechanism for multicore systems with global EDF
scheduling \cite{pellizzoni2008m}. It uses CASH for periodic tasks,
and an utilization based algorithm similar to GRUB for sporadic tasks
and for reclaiming the extra free bandwidth in the system. 

The authors of \cite{abeni2016multicore} propose two different
reclaiming mechanism for global EDF based on GRUB: parallel reclaiming
is a global reclamation scheme (similar to the sporadic mechanism
proposed in M-CASH) where all available unused bandwidth is stored in
a single global variable accessible by all cores. Sequential
reclaiming removes this constraint by storing the unused bandwidth on
a per-core basis. Unfortunately, due to the global nature of the
scheduler, both M-CASH, Parallel and Sequential reclaiming schemes
suffer from a relatively large amount of task migrations.



\section{System model}
\label{sec:sys-model}

A real-time task $\task{i}$ is a (possibly infinite) sequence of jobs
$\job{i}{k}$. Each job $\job{i}{k}$ arrives at time $\arrival{i}{k}$
and is expected to complete its execution before time
$\arrival{i}{k} + \deadline{i}$, where $\deadline{i}$ is the
\emph{relative deadline} of the task. A job \emph{misses its
  deadline} if the job finishing time $\finishtime{i}{k}$ is greater
than $\arrival{i}{k} + \deadline{i}$. A task is \emph{periodic} if the
arrival time between two consecutive jobs is fixed to $\period{i}$,
called period ($\arrival{i}{k} = \arrival{i}{0}+ K \cdot \period{i}$).
Sporadic tasks relax the last constraint so the arrival time between
two consecutive jobs is equal or greater than the task's period
$\period{i}$ ($\arrival{i}{k+1} - \arrival{i}{k} \geq \period{i}$).
In this paper we assume that all tasks have deadline equal to their
period (or their minimum inter-arrival time). Let $\charge{i}{k}$
denote the execution time of job $\job{i}{k}$: its exact value is not
know before the execution of the job.

A \emph{server} is a scheduling abstraction that is used by the
scheduler to provide temporal isolation between different tasks. In
the resource reservation framework, each task is assigned a server
$\server{i}$ characterized by a budget $\budget{i}$ and a period
$\serverp{i}$. The resource reservation algorithm guarantees that the
task can execute for a minimum amount of time $\budget{i}$ every
period $\serverp{i}$.  The ratio $\budget{i} / \serverp{i}$ is the
\emph{server utilization} $\serveru{i}$: it represents the fraction of
CPU time reserved to $\task{i}$. The total CPU utilization
$\bandwith{}$ of a set of servers is simply the sum of all server
utilization's $\bandwith{} = \sum_{i=1}^n \serveru{i}$.  When using
partitioned scheduling, tasks and servers are allocated to CPUs/cores;
we denote by $\taskset{j}$ the set of tasks allocated on core $j$, and
$\bandwith{j}$ the total utilization of $\taskset{j}$.

In this paper we consider an
identical multicore platform: all cores have the same characteristics
(architecture, micro-architecture, $\ldots$) and share the same fixed
operating frequency.

\section{The GRUB algorithm}
\label{sec:grub-algorithm}

In this section, we recall the GRUB (Greedy Reclamation of Unused
Bandwidth) algorithm~\cite{lipari2000greedy} and its main
properties. Therefore, in this section we restrict to single processor
scheduling.
 
Let $\task{i}$ be a task of $\taskset{}$. All jobs of task $\task{i}$
are executed according to their arriving order: job $\job{i}{k}$
cannot start executing until job $\job{i}{k-1}$ has finished. Each
task is assigned to a server $\server{i}$ which is characterized by
two parameters: the server bandwidth $\serveru{i}$ and the server
period $\serverp{i}$.

The run-time behavior of each server $\server{i}$ is described by two
state variables, that are computed and updated at run-time: the server
deadline $\serverd{i}$ and the virtual time $\virtualtime{i}$.  The
urgency of each job of $\task{i}$ depends on its server's deadline,
hence jobs of $\taskset{}$ are scheduled according to their earliest
server's deadline (EDF-like policy). Please notice that the task's
deadline $\deadline{i}$ and the server's deadline $\serverd{i}$ are
two different entities ($\deadline{i}$ is a {\em relative} deadline,
while $\serverd{i}$ is an {\em absolute} deadline). Also, please
notice that the absolute deadline of job $\job{i}{k}$ is
$\arrival{i}{k} + \deadline{i}$ and it may or may not coincide with
the server's deadline $\serverd{i}$. A global variable
$\bandwith{}^{a}$ stores the current \emph{active load} on the
processor. $\bandwith{}^{a}$ is initialized to 0.

The GRUB algorithm guarantees that all the servers' deadlines are
respected (each server $\server{i}$ allows its task $\task{i}$ to
execute for $\budget{i}$ time units before $\serverd{i}$) if Equation
\eqref{eq:totalbandwith} (EDF schedulability) is respected:
\begin{equation}
  \label{eq:totalbandwith}
  \sum_i \serveru{i} \leq 1.
\end{equation}

The dynamic evolution of $\server{i}$ is described by the finite state
machine shown in Figure \ref{fig:statemachine}.

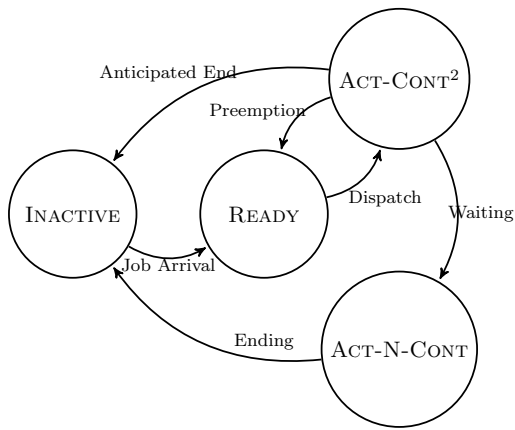
\begin{figure}[tb]
  \centering
  \resizebox{0.65\columnwidth}{!}{
    \begin{tikzpicture}[->,>=stealth',shorten >=1pt,auto,node distance=3cm,
      thick,main node/.style={circle,draw}, minimum size = 2cm]
      \node[main node] (1) {\grubstate{Inactive} };
      \node[main node] (2) [right of=1] {\grubstate{Ready}};
      \node[main node] (3) [above right of=2] {\grubstate{Act-Cont}\footnote{active-contending}};
      \node[main node] (4) [below right of=2] {\grubstate{Act-N-Cont}};
      
      \path[every node/.style={font=\sffamily\small}]
      (1) edge [bend right]node{}(2)
      (3) edge  [bend right] node {}(1)
      (2) edge  [bend right] node {}(3)
      (3) edge  [bend right] node {}(2)
      (3) edge  [bend left] node {}(4)
      (4) edge [bend left] node {}(1)
      ;
      
      \node at (1.5,-0.8){\footnotesize Job Arrival};
      \node at (4.9,0.25){\footnotesize Dispatch};
      \node at (2.9,1.6){\footnotesize Preemption};
      \node at (6.4,0){\footnotesize Waiting};
      \node at (3,-2){\footnotesize Ending};
      \node at (1.5,2.2){\footnotesize Anticipated End};
    \end{tikzpicture}
  }
  \caption{Grub Control automaton}\label{fig:statemachine}
\end{figure}

\paragraph{\grubstate{Ready} state}
Server $\server{i}$ is in \grubstate{Ready} state when there is at least a
pending job and server's deadline $\serverd{i}$ is not the earliest.

When the server moves from \grubstate{Inactive} to \grubstate{Ready} because a
new job has arrived at time $t$, the server's virtual time
$\virtualtime{i}$ and deadline $\serverd{i}$ are updated as follows:
\begin{align}
  \virtualtime{i}& \leftarrow \arrival{i}{k} \label{eq:varrival}\\ 
  \serverd{i}& \leftarrow \virtualtime{i} + \serverp{i} 
\end{align}
Also, its utilization must be added to the utilization of the active
tasks:
\begin{align}
   \bandwith{}^{a} \leftarrow \bandwith{}^a + \serveru{i}
\end{align}
where $\bandwith{}^{a}$ is the total utilization of the active
servers. 

A server $\server{i}$ in \grubstate{Ready} state goes into
\grubstate{Executing} state when its deadline is the earliest among all
ready servers.

\paragraph{\grubstate{Executing} state}
While server $\server{i}$ is in \grubstate{Executing} state for $\Delta t$
units of time, its virtual time $\virtualtime{i}$ is updated as
follows:
 \begin{equation}
   \virtualtime{i}(t+\Delta t) \leftarrow  \virtualtime{i}(t) + \frac{\bandwith{}^{a}}{\serveru{i}} \Delta t
 \end{equation}
 If, while in \grubstate{Executing}, the server's deadline becomes equal to
 its virtual time, the deadline is postponed to:
 \begin{equation} \label{eq:incrementdeadline}
   \serverd{i} \leftarrow \virtualtime{i}(t) + \serverp{i}
 \end{equation}

 Notice that deadline postponing can decrease the relative priority of
 the server, that can then be preempted by other more urgent servers
 that are present in the ready state.

 When a job ends its execution and there are no more pending jobs of
 the same task, its server can move out of the \grubstate{Executing}
 state: if its virtual time is greater than $t$, it moves to
 \grubstate{Act-N-Cont} state; otherwise, it goes to \grubstate{Inactive}.

\paragraph{\grubstate{Act-N-Cont} state}
When the server is in \grubstate{Act-N-Cont} state, even if no job is
pending we still consider it as \emph{active}, because its virtual
time is in the future, which means that it has already consumed all
its reserved bandwidth until time $\virtualtime{i}{t}$. Hence, its
utilization is still accounted for in the active utilization $\bandwith{}^a$.

In this state, the server does not update its variables: when $t$
reaches $\virtualtime{i}$, the server moves to the \grubstate{Inactive}
state.

If a new job arrives while the server is in \grubstate{Act-N-Cont} state,
then the server moves to state \grubstate{Ready} without updating its
variables.

\paragraph{\grubstate{Inactive} state}
In this state, the server has no pending job to serve.  Every time the
server enters this state, the active bandwidth is updated accordingly:
\begin{equation}
  \bandwith{}^a  \leftarrow \bandwith{}^a - \serveru{i}
\end{equation}


\subsection{Properties}
\label{sec:grub_properties}

In this section we recall some of the properties of the GRUB server
that will be useful later for our proof of correctness. We start by
defining the notion of \emph{active server}. A server is active if it
is in any of the states \grubstate{Ready}, \grubstate{Active-Contending},
\grubstate{Active-Non-Contending}. Only active servers contribute with
their bandwidth to variable $\bandwith{}^a$ and hence can impact the
amount of available reclaiming. The following property holds.

\begin{mylemma}
  \label{th:grub_active_utilisation}
  Given a single processor system scheduled by the GRUB algorithm, and
  let $\bandwith{}^a(t)$ be the sum of the bandwidth of all
  \emph{active} servers in the system at time $t$.

  If $\forall t \;\; \bandwith{}^a(t) \leq 1$, then no server misses
  its scheduling deadline, regardless of the behavior of the served
  tasks.
\end{mylemma}
\begin{proof}
  The proof can be found in \cite{lip}.
\end{proof}

Notice that Lemma \ref{th:grub_active_utilisation} does not require a
static set of servers: a server can enter the system at any time, and
it will not impact the guarantees on the already existing servers, as
long at the total active bandwidth does not exceed $1$ at any instant.

A consequence of the previous Lemma is that, to guarantee that no
server ever misses its deadline, we can limit the sum of the
bandwidths of all servers currently in the system, whether active or
not.
\begin{mycollo}
  If the sum of the utilization of all the servers currently assigned
  to the processor does not exceed 1, no server deadline is missed.
\end{mycollo}
\begin{proof}
  Since $\forall t \;\; \sum_i \bandwith{i} \leq \bandwith{}^a$, the
  corollary simply descends from Lemma
  \ref{th:grub_active_utilisation}.
\end{proof}

\subsection{Advantages of GRUB}
\label{sec:advantages-grub}

Compared to other reclaiming algorithms
(e.g. \cite{Caccamo05,lin2005improving}), GRUB has some advantages:
\begin{itemize}
\item It can be used with periodic, sporadic and aperiodic
  non-real-time tasks;
\item It automatically reclaims the unreserved bandwidth;
\item It has the same complexity as the CBS algorithm
  \cite{abeni1998integrating};
\item It keeps track of the \emph{active bandwidth} on the processor;
  this property will be used for \emph{temporary migration} (see
  Section \ref{sec:grub-tempmigr});
\item It can be used to lower the processor frequency (DVFS) in order
  to reduce energy consumption \cite{journals/tc/ScordinoL06}.
\end{itemize}

The GRUB algorithm has already been extended to multicore systems. In
\cite{pellizzoni2008m}, an algorithm similar to GRUB is used to
reclaim global unused bandwidth in the system; in
\cite{abeni2016multicore}, \emph{parallel reclaiming} and
\emph{sequential reclaiming} strategies based on GRUB have been
proposed. In both cases, the underlying scheduling algorithm is global
EDF. In this paper, we address partitioned EDF scheduling.

\section{GRUB and Migrations}
\label{sec:grub-tempmigr}

From now on, we assume a platform consisting of $m$ identical CPUs.
We propose to partition a set of GRUB servers into the $m$ available
CPUs. Each CPU has its own ready queue and a single-processor
scheduler based on EDF. Servers are partitioned using one of the
classical bin-packing heuristics available in the literature.

We assume that tasks (and the corresponding servers) may enter and
leave the system at any time. Suppose a new task $\task{i}$ enters the
system at time $t$, with server parameters
$\server{i} = (\budget{i}, \serverp{i})$. Then, the heuristic
allocation algorithm is run to select the CPU on which the server will
be allocated. The heuristic uses an \emph{admission control} test to
find the most suitable core where to allocate the server. The
allocation algorithm will be described in
Section~\ref{sec:allocation-strategy}.

After partitioning, tasks are scheduled on each CPU using the GRUB
algorithm. Additionally, tasks can temporarily migrate to other CPUs
to reclaim extra bandwidth. For every CPU, the algorithm maintains the
following variables:
\begin{itemize}
\item $\bandwith{j}$ is the total utilization of the servers allocated
  on core $j$; this \emph{does not} include the tasks that have been
  temporarily migrated on $j$.
\item $\bandwith{j}^m(t)$ is the total bandwidth at time $t$ of the
  tasks that have been temporarily migrated on core $j$.
\item $\bandwith{j}^a(t)$ is the total \emph{active} utilization of
  all the active servers on core $j$ at time $t$; this may include the
  bandwidth of the active tasks that have temporarily been migrated on
  $j$.
\end{itemize}

Each server is assigned an additional parameter, the \emph{migrating
  utilization} $\serveru{i}^m$. The migrating utilization is used to
distribute the unused bandwidth on the destination core to the
(possibly many) incoming servers. As we will see in the next section,
by setting its migrating utilization to $0$, temporary migration is
disabled for a given task.

\subsection{Temporary Migrations}
\label{sec:job-migration}

Suppose that task $\task{i}$ is allocated on core $j$ and it is served
by $\server{i}$. Consider job $\job{i}{k}$: it starts executing on
core $j$ according to the original GRUB algorithm. Suppose that at
some time $t_0$ the virtual time $\virtualtime{i}(t_0)$ of the server
becomes equal to the server's deadline $\serverd{i}$: the job has
consumed all its budget and can not reclaim the unused bandwidth of
the other tasks allocated to core $j$.

In the original GRUB algorithm, the server deadline is postponed; in
the partitioned version the job is declared as \emph{eligible 
for migration}. This gives $\job{i}{k}$ a chance to continue its
execution with the same server deadline (same urgency) on a different
core $j'$. In the following discussion, we denote core $j'$ as the
\emph{destination core}.

A job is eligible for migration at time $t_0$ when the following two conditions are verified:
\begin{align}
  \virtualtime{i}(t_0)& \geq \serverd{i}   \label{eq:migration} \\
  \serverd{i} &> t_0                       \label{eq:migration2}
\end{align}
If \eqref{eq:migration} is true, task $\task{i}$ cannot continue to
execute on its current core without postponing the current server
deadline $\serverd{i}$, hence we may try to migrate it to a different
core. If Condition~\eqref{eq:migration2} is not verified, the current
server deadline $\serverd{i}$ has been reached, so there is no point
in migrating the task.

When a job is eligible for migration, the algorithm first selects the
destination core $j'$ as the one with the smallest active utilization
$\bandwith{j'}^a$, because on $j'$ there is more chance for the task
to reclaim extra bandwidth.

Then we must check that the migrating job fits in the
destination core with its migrating utilization:
\begin{equation}
  \label{eq:job-admission}
  \serveru{i}^m + \bandwith{j'}^m(t) + \bandwith{j'} \leq 1
\end{equation}
Recall that bandwidth $\bandwith{j'}$ is the total bandwidth of all
servers that are allocated on $j'$, whereas bandwidth
$\bandwith{j'}^m(t)$ is the total active bandwidth of the jobs that
have temporarily been migrated to $j'$ at time $t$. In practice,
Equation \eqref{eq:job-admission} guarantees that we are not
overloading the destination core with too many migrating jobs.

If Equation~\eqref{eq:job-admission} is not respected for the selected
core $j'$, we can still try to migrate the job by reducing its
migrating utilization $\serveru{i}^m$ to: 
\begin{equation} \label{eq:um} 
  \serveru{i}^{m\prime} = \min \{ \serveru{i}^{m}, 1-(\bandwith{j'} + \bandwith{j'}^m(t) \}
\end{equation}

The task is migrated by creating a new server $\server{i}^\prime$ on
the destination core $j'$ with:
\begin{itemize}
\item the same server period $\serverp{i}$, 
\item utilization equal to $\serveru{i}^{m\prime}$ as computed by
  Equation~\eqref{eq:um};
\item the server state is initialized to \grubstate{Ready};
\item its virtual time is initialized to $t_0$ and the server deadline
  is the same as $\serverd{i}$.
\end{itemize}
Task $\task{i}$ is served by the new temporary server
$\server{i}^\prime$ until the current job $\job{i}{k}$ completes. This
temporary server is managed by the GRUB algorithm like a regular
server.

Once the job completes, the task returns immediately to its original
processor and the temporary migration is concluded: the next jobs of
the task will start executing on the original processor. However, the
temporary server $\server{i}^\prime$ is not immediately deleted: when
the job completes, it follows the rules of the GRUB algorithm, and if
$\virtualtime{i}(t) > t$, the server first moves to state
$\grubstate{Act-N-Cont}$; later, when $\virtualtime{i}(t) = t$, it moves
to $\grubstate{Inactive}$ and it can be deleted from the system. Notice
that a temporary server remains active during its lifespan.

Bandwidth $\bandwith{j'}^m(t)$ 
keeps track of the total bandwidth of all tasks that are temporarily
migrated to core $j'$: it is incremented by $\serveru{i}^{m\prime}$
when the temporary server is created, and it is decremented by the
same amount when the temporary server becomes $\grubstate{Inactive}$ and
it is deleted.

Algorithm \ref{alg:jobmigration} modifies the original GRUB algorithm
to take into account job migration. When $\virtualtime{i}$ reaches
$\serverd{i}$ (Line 3 of the algorithm, implementing
Condition~\ref{eq:migration}), instead of simply postponing the server
deadline as in the original GRUB, Algorithm~\ref{alg:jobmigration}
tries to migrate the task by first selecting a destination core (Line
4). The algorithm uses a boolean flag per each task, denoted as
\texttt{migrated}$_{i}$, initialized to \texttt{false}. If the task
has not yet been migrated (condition at Line 5), to ensure that the
selected core has enough bandwidth to accommodate for the incoming
task, we update the migrating utilization (Line 7).
At Line 8, we further limit the number of migrations by using an
appropriate threshold $\epsilon$ to guarantee that the task will be
able to execute enough time on the target core, otherwise the overhead
of migration can overcome its benefits. 

\begin{algorithm}[t]
  \caption{Temporary Migration}\label{alg:jobmigration}
  \begin{algorithmic}[1]

  \State \emph{Classic\_Grub\_scheduling\_Code}
  \State ~~~~~~~~~~~~~ $\cdots$
  \If {($\virtualtime{i} \geq \serverd{i}$)}
        \State $\textsc{p}^m$ = selectDestinationCore()
        \If {(\textbf{not} \texttt{migrated}$_i$)} 
            \State $\bandwith{j} = \bandwith{\textsc{p}^m} + \sum \serveru{a}^m $
            \State $\serveru{i}^m = \min(\serveru{i}^m, 1 - \bandwith{j})$
            \If {($\serveru{i}^m (\serverd{i} - t) / (\serveru{i}^m+U^{a}_m) >
\epsilon$)}
                \State Create temporary server $\server{i}' = (\serveru{i}^m, \serverp{i})$
                \State Assign task $\task{i}$ to $\server{i}'$
                \State \texttt{migrated}$_i$ = true 
                \State $ \serverd{i}' = \serverd{i}$
                \State $\virtualtime{i}'(t) = t$
                \State State of $\server{i}'$ = \grubstate{Ready};
                \State DoMigration();
             \Else
                \State \texttt{migrated}$_i$ = false 
                \State $\serverd{i} = \virtualtime{i}(t) + \serverp{i}$
             \EndIf
        \Else  
          \State \texttt{migrated}$_i$ = false 
          \State $\serverd{i} = \virtualtime{i}(t) + \serverp{i}$
        \EndIf
      \EndIf
  \State ~~~~~~~~~~~~~ $\cdots$
  \State \emph{Classic\_Grub\_scheduling\_Code} 
\end{algorithmic}
\end{algorithm}

If the current job has already been migrated ({\tt migrated}$_i$ is
true) or if it has not the possibility to execute for enough time on
the target core (the test at Line 8 is false), then the server
deadline is postponed (Lines 19 and 22). The migration flag is set to
limit the number of possible migrations of the same job, and it is
reset to \texttt{false} once the job completes.

\subsection{Proofs of correctness}
\label{sec:proofs-correctness}

In this section we discuss the correctness of our partitioned GRUB
algorithm. In particular, we prove that the temporal isolation
property is still valid when we introduce the temporary migration
mechanism.

We start by observing that, on each core, we execute an instance of
the GRUB algorithm independent of the others. Therefore, as long as no
migration is allowed, we can use Equation~\eqref{eq:totalbandwith} as
an admission control for every core. Let us now consider the case of a
temporary migration.
\begin{mylemma}
  \label{th:active-utilisation-after-migration}
  Assume that task $\tau_i$ is migrated from core $j$ to core $j'$ at
  time $t_0$. Then, 
  \begin{align*}
    \forall t > t_0, \;\; \bandwith{j}^a(t) \leq 1 \;\; \wedge \;\; \bandwith{j'}^a(t) \leq 1
  \end{align*}
\end{mylemma}
\begin{proof}
  Assume that, before migration, the active bandwidths on core $j$ and
  $j'$ are not greater than 1:
  \begin{align*}
  \forall t < t_0, \;\; \bandwith{j}^a(t) \leq 1 \;\; \wedge \;\; \bandwith{j'}^a(t) \leq 1
  \end{align*}
  This property is clearly true before the first migration: we want to
  prove that it remains true after each migration.

  The lemma is trivially true on core $j$: in fact, the original
  server of the task is not deleted: it remains initially in
  \grubstate{Active-Non-Contending} and then it may move to
  \grubstate{Inactive}: in any case, the active bandwidth on core $j$
  cannot not increase due to the migration.
  
  Regarding core $j'$, by definition:
  \begin{equation}\label{eq:cap-on-active-bandwidth}
    \forall t < t_0, \;\;  \bandwith{j'}^a(t) \leq \bandwith{j'} + \bandwith{j'}^m(t)
  \end{equation}
  i.e. the total active bandwidth does not exceed the sum of the
  bandwidth of the tasks allocated on $j'$ and the bandwidth of the
  temporary servers on $j'$. 
  
  At time $t_0$, a \emph{temporary server} is created on core $j'$
  with server utilization calculated according to
  Equation~\eqref{eq:um}. Notice that the temporary server is active
  for the duration of its lifespan, and hence its utilization is
  immediately summed to the active utilization on core $j'$. Let us
  denote by $t_0^-$ the instant before the migration takes place:
  \begin{equation}\nonumber
    \bandwith{j'}^a(t_0) =  \bandwith{j'}^a(t_0^-) + \serveru{i}^{m\prime} \leq \bandwith{j'}^a(t_0^-) + 1 - (\bandwith{j'} + \bandwith{j'}^m(t_0^-)) \leq 1,
  \end{equation}
  the last inequality is verified by substituting
  Equation~\eqref{eq:cap-on-active-bandwidth}.

  Furthermore,
  $\bandwith{j'}^m(t_0) = \bandwith{j'}^m(t_0^-) +
  \serveru{i}^{m\prime}$,
  and $\bandwith{j'} + \bandwith{j'}^m(t_0) \leq 1$. By definition of
  the GRUB algorithm,
  $\bandwith{j'}^a(t) \leq \bandwith{j'} + \bandwith{j'}^m(t)$, so the
  active bandwidth cannot exceed 1 after the migration.  

  Finally, notice that, once the migrated job has finished, the server
  is deleted only when its state becomes $\grubstate{Inactive}$: at that
  point, its bandwidth is subtracted both from $\bandwith{j'}^m(t)$
  and from $\bandwith{j'}^a(t)$, so
  $\bandwith{j'}^a(t) \leq \bandwith{j'} + \bandwith{j'}^m(t)$ even
  after the migration is over.

  Hence the lemma is proved.
\end{proof}

\begin{mytheo}
  \label{th:migration-correct}
  Let us assume a multicore platform with $m$ cores, and a set of
  servers partitioned on the $m$ cores such that:
  \[
  \forall j = 1, \ldots, m  \quad \bandwith{j} \leq 1
  \]
  Then, when scheduled by partitioned GRUB with temporary migration,
  no server misses its scheduling deadline.
\end{mytheo}
\begin{proof}
  It descends from Lemma \ref{th:active-utilisation-after-migration}
  and from Lemma \ref{th:grub_active_utilisation}.
\end{proof}

\subsection{Short server periods}
\label{sec:short-server-periods}


In general, job deadline $\arrival{i}{k}+\deadline{i}$ can be
different than the server deadline $\serverd{i}$. In fact, in many
cases it is useful to set the server period $\serverp{i}$ equal to a
divisor of the task's period $\period{i}$,
$\period{i} = k \serverp{i}$. For example, in the adaptive reservation
framework described in \cite{abeni2002analysis,abeni2005qos} the
authors suggest to use such technique to improve the performances of
the adaptive control mechanism.

In this case, when the virtual time $\virtualtime{i}(t_0)$ reaches the
current server deadline $\serverd{i}$, the job deadline
$\arrival{i}{k}+\deadline{i}$ can still be far in the future and a
migration may be not strictly needed: the job has still a chances to
meet its deadline without migrating. 

For the sake of simplicity in this paper we only describe the case
when the server's period is equal to the task period. However, it is
easy to modify our algorithm to account for this case by adding an
additional server parameter \textsc{MinPost} that imposes a minimum
number of deadline postponing before making the task eligible for
migration. This will be the subject of a future work.

\section{Permanent Migrations}
\label{sec:grub-partitioning}

Temporary migration is useful for reclaiming extra bandwidth. However,
it may happen that a task still suffers too many deadline misses even
with temporary migration enabled. When a task consistently misses too
many deadlines, it can be useful to permanently migrate it to a different
core, using some kind of load balancing mechanism. More specifically,
if more than a specified amount of jobs have missed their
deadlines in an interval of time of size $\mathsf{w}_i$, it may be
necessary to \emph{permanently} migrate the task in the hope to
improve its quality of service.

In contrast to temporary job migration which happens while the job is
running, the condition for permanent task migration is checked at the
end of the execution of a job. If a task is permanently
migrated, all its future jobs will start executing on the new processor.

First, we select a destination core for the task. We distinguish two
cases:
\begin{itemize}
\item The task seeks to migrate immediately, thus the migration core
  must verify the following condition:
  \begin{equation}\nonumber
    \serveru{i} \leq 1 - \bandwith{j} - \bandwith{j}^m
  \end{equation}
  If this is verified, we can immediately migrate the task with its
  server in the new core (see below).

\item If no processor guarantees the above condition, then we can try
  to delay the migration to some time in the next future. In this
  case, we first look for a processor in which:
  \[
     \serveru{i} \leq 1 - \bandwith{j}
  \]
  Then, we disable incoming migrations on the destination core, and we
  wait for $\bandwith{j}^m \leq 1 - (\serveru{i} + \bandwith{j})$. We
  have to wait at most for the longest response time of any migrating
  job: this depends on the job execution time and cannot be easily
  bounded a-priori. Therefore, we additionally start a timer, and if
  the migration bandwidth has not decreased enough within the timeout,
  we abort the load balancing operation.

\item If none of the conditions above is verified, permanent task
  migration is not possible.
\end{itemize}

If permanent migration is possible, the task and the corresponding
server are migrated on the destination core $j'$ with their own
variables, and $\bandwith{j'}$ is incremented. We can decrease the
$\bandwith{j}$ on the original core only once the original server
becomes inactive (i.e. the server virtual time $\virtualtime{i} =
t$). This can be easily achieved by setting an appropriate timer.

If multiple target cores are available as destination for permanent
migrations, some ``classical'' heuristics such as Best Fit (BF),
First Fit (FF) and Worst Fit (WF) can be used to select the destination
core.
The same heuristics can be used for the {\em allocation strategy} presented
in the next section, and will be used for the simulations reported in
Section~\ref{sec:results-discussions}.

The proof of correctness of this mechanism is similar to the proof of
Corollary \ref{th:migration-correct}, and it is not reported here because of
space constraints.

\subsection{Allocation strategy}
\label{sec:allocation-strategy}

When a new task $\task{n}$ enters the system, the allocation strategy
used to schedule it on the most suitable core is similar to the strategy
used for load balancing. We first try to
immediately allocate the task on one of the cores were
$1-(\bandwith{j} + \bandwith{j}^m) \geq \serveru{n}$. If no core meets
the condition, we select one of the cores where
$1-\bandwith{j} \geq \serveru{n}$, we disable migration onto the
selected core and we wait for the migrating utilization to decrease.
If more than one possible target core is found, we select one according
to one of the heuristics mentioned above (BF, FF, or WF).
When a task leaves the system, its servers utilization is subtracted
from the processor utilization only when the server becomes
\grubstate{Inactive}.

If none of the cores guarantees that
$1-\bandwith{j} \geq \serveru{n}$, or if the maximum timeout for
decreasing the migrating utilization expires, the incoming task is
rejected.
In this case, the user may decide to 
assign the task a smaller $\serveru{n}$, of course losing in quality
of service, or to kill some other running task to make place for the
new one.


\section{Experimental Evaluation}
\label{sec:results-discussions}

The performance of the proposed solution have been evaluated
through an extensive set of simulations, and the most interesting results
are reported in this section.
The simulations have been performed by using a discrete time simulator
built in SCALA that adopts a functional programming paradigm.
The simulator supports both partitioned and global scheduling and can simulate
both sporadic and periodic tasks. It implements the GRUB reclaiming algorithm
in all its known incarnations: single-core~\cite{lipari2000greedy},
partitioned (as proposed in this paper, using various heuristics for
load balancing and tasks allocation)
and global~\cite{abeni2016multicore} (either with sequential or parallel
reclaiming). Global GRUB with sequential and parallel reclaiming will
be referred as ``G-Seq'' and ``G-Par'' in the following.

Thanks to the usage of the functional programming paradigm and to its modular
design, the simulator code can be easily extended with new reclaiming policies.

The simulator will be soon available online at
\url{https://github.com/zahoussem/reclaiming_simulator}

\subsection{Experimental Setup}

Each simulation scenario consists of randomly generated task
sets. First we fix the number of tasks to generate to a fixed number
$n$ and the number of cores $m$, and we set a target total utilization
level. Then we generate a set of $n$ servers' utilization
$\{ \serveru{i} \}$ by using the UUNIFAST-discard algorithm
\cite{emberson2010techniques}.

Then, we generate the execution times of the jobs. Execution times are
generated according to a certain probability distribution between a
minimum and a maximum value. First we randomly select a minimum
execution time $\mathsf{minexec}$ and a maximum execution time
$\mathsf{maxexec}$ for each task as random samples from a random
variable $\mathsf{Unif(5, 200)}$ with uniform distribution. Then a
sequence of job execution time is generated from a random variable
between $[\mathsf{minexec}, \mathsf{maxexec}]$. In this paper we
considered two different distributions: the two-level uniform
distribution and the Weibull distribution.

Two-level uniform distribution takes as an additional parameter a
desired budget $B$ and a probability $\mathsf{PM}$ (set to $0.75$ in
the experiments). Then, the execution time is chosen with an uniform
distribution between $[\mathsf{minexec}, B]$ with probability
$\mathsf{PM}$, and with an uniform distribution between
$[B+1, \mathsf{maxexec}]$ with probability $1 - \mathsf{PM}$. In this
way we can precisely and easily control the number of jobs with
execution time larger than $B$.  However, the two-level uniform
distribution does not represent realistic workloads.

The Weibull distribution is known to be a good model of the execution
time of many real-time workloads~\cite{Cucu2012} because it
effectively represents the statistical behavior of typical execution
time profiles with long tails. The Weibull probability density
function can be expressed as:
\[
   f(x) = \frac{k}{\lambda} \left(\frac{x-\theta}{\lambda}\right)^{k-1} e^{-\left(\frac{x-\theta}{\lambda}\right)^{k}}
\]
where $k$ is the \emph{shape} parameter, $\lambda$ is the \emph{scale}
parameter, and $\theta$ is the location (or displacement) parameter.

Once execution times are generated, the budget $B$ for the task is
selected so to ensure that approximately a certain percentage
$\mathsf{PM}$ of jobs will have execution time greater than $B$. In
both cases, the task period and the server period are computed as
$B / \serveru{i}$. For every server, the migrating utilization has
been set to $0.1$.



In the simulation experiments that we performed, we did not observe
statistically relevant differences between the two-levels uniform
distribution and the Weibull distribution for the migration ratio and
the deadline miss rate ratio. Therefore, in this paper we discuss only
the results obtained using the two-level distribution.

Once the set of server is generated, we test schedulability and
allocation. In particular, we test that the set can be partitioned
using FF, BF and WF; and we test the schedulability of the servers
with the schedulability tests for global EDF described in
\cite{abeni2016multicore}. If any of the allocation algorithms fails,
or any of the scheduling tests fails, we discard the generated task
set.

In the rest of this section, the results of simulations are presented
and discussed. We compare the performances of different partitioning
heuristics and the parallel and sequential reclaiming techniques
described in \cite{abeni2016multicore} according to the number of
deadlines missed and the number of job and task
migrations. 

\subsection{Impact of Temporary Migrations}
\label{sec:comp-job-migr}

To evaluate the effects of job migration, we set the number of cores
to $m=4$ and the number of tasks to $n= 25$. The total utilization
ranges between $[0.5, 3]$. In fact, it is very difficult to generate
task sets with high total utilization that are schedulable by all the
scheduling methods considered in this comparison. For each utilization
level, $100$ scenarios have been generated. For our partitioned
reclaiming scheme, in the figures presented in this section job
migration is enabled and load balancing is disabled.

\begin{figure}[tb]
  \centering
 \begin{tikzpicture}[scale=0.8]
    \begin{axis}
      [ axis x line=bottom, axis y line = left, xlabel={Total
        Bandwidth}, ylabel={average ratio of migrations}, ymin = 0,
      legend entries={"G-Seq","G-Par","FF","BF","WF"}, legend
      style={at={(1,0.7)}}, ]

      \addplot table[y index=1,x index=0]{figs/mig_job.txt};
      \addplot table[y index=2,x index=0]{figs/mig_job.txt};
      \addplot table[y index=3,x index=0]{figs/mig_job.txt};
      \addplot table[y index=5,x index=0]{figs/mig_job.txt};
      \addplot table[y index=4,x index=0]{figs/mig_job.txt};
    \end{axis}
  \end{tikzpicture}
  \caption{Migrations per job as a function of the system load.}
  \label{fig:res:ad}
\end{figure}
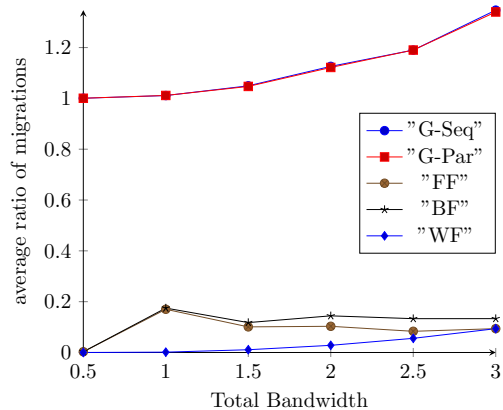

Figure \ref{fig:res:ad} shows the average number of migrations
per job as a function of total
utilization. The figure compares the results of global parallel
reclaiming (G-Par), global sequential reclaiming (G-Seq), and our
reclaiming scheme with Worst-Fit (WF), First-Fit (FF) and Best-Fit
(BF). As expected, the number of migrations in global algorithms G-Par
and G-Seq is much larger compared to partitioned schemes.

Notice that for utilization $0.5$, the reclaiming algorithm is able to
respect the tasks constraints without postponing the scheduling deadlines;
hence, the partitioned algorithms do not produce any migration. When the load
increases, BF and FF exhibit a different behavior compared to  WF: 
FF and BF tend to allocate the maximum number of tasks on the least
possible number of cores, so this may lead to cores that are very
packed while other cores are less loaded. In contrast, WF tends to
spread the workload evenly across cores. Thus, in BF and FF many jobs
have less opportunities to reclaim the local unused bandwidth and are
more eligible to migrate compared to WF.

For example, when the utilization is $1.0$, FF and BF allocate all the
jobs on the first core, leaving all the other cores completely
unloaded. Hence, all the tasks will migrate to reclaim CPU time from
the unloaded cores. This is why for utilization $1.0$ BF and FF show a
large increment in the number of migrations per job.  When the load
increases, BF and FF start to assign tasks to the second core (then to
the third and to the fourth). Since the number of tasks is fixed, when
increasing the total utilization, the average utilization of every
single task also increases, so BF and FF tend to distribute the load
more uniformly. As a result, the number of migrations per job
decreases and for utilization $3.0$ the numbers of migrations for BF,
FF and WF are similar.

\begin{figure}[tb]
  \centering
 \begin{tikzpicture}[scale=0.8]
    \begin{axis}
      [ axis x line=bottom, axis y line = left, xlabel={Total
        Bandwidth}, ylabel={average ratio of missed deadlines},
	ymin = 0,
	legend entries={"G-Seq","G-Par","FF","BF","WF"},
	legend style={at={(0.5,0.9)}},
yticklabel style={
        /pgf/number format/fixed,
        /pgf/number format/precision=5
},
scaled y ticks=false,
      ]
      \addplot table[y index=4,x index=0]{figs/miss_job.txt};
      \addplot table[y index=5,x index=0]{figs/miss_job.txt};
       \addplot table[y index=1,x index=0]{figs/miss_job.txt};
      \addplot table[y index=3,x index=0]{figs/miss_job.txt};
     
      \addplot table[y index=2,x index=0]{figs/miss_job.txt};
    \end{axis}
  \end{tikzpicture}
  \caption{Average missed deadlines ratio.}
  \label{fig:res:miss_ad}
\end{figure}
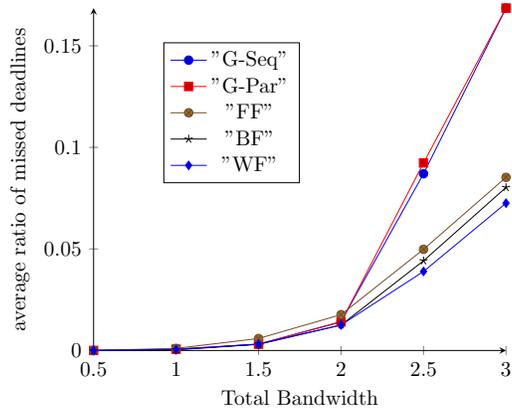

For the same simulation experiment, Figure \ref{fig:res:miss_ad} shows
the average ratio between number of deadline misses and total number
of jobs as a function of total utilization.

At low total utilization, all scheduling techniques allow all tasks to
respect their deadlines. The higher the total bandwidth, the higher is
the probability of missing the deadline.  However, partitioned schemes
present significantly better performance.
%
While all of the algorithms show
almost zero deadline miss for utilizations up to $1.0$, the partitioned
schedulers cause much less migrations (in particular, WF causes $0$ migrations
up to utilization $1.0$).
This shows that with a proper partitioning the local reclaiming rule of the GRUB algorithm
can reclaim enough bandwidth for all jobs to complete
within their deadlines.
Moreover, when the utilization grows above $2.0$ the number of deadline
misses with the two global schedulers is much higher than with the
partitioned schedulers.

When comparing partitioned scheme against each other, again WF shows a
slightly better behavior, because it can more effectively balance the
load across all processor.

The confidence interval for partitioned techniques is less than
$\pm 0.0025$ and less than $\pm 0.0047$ for the global
techniques. Similar values of the confidence interval has been noticed
for all figures presented in this section.

\subsection{Impact of Permanent Migrations}
\label{sec:load-balancing}

In this section, we evaluate the performances of our heuristics when
both job migration and load balancing are activated. Load balancing is
useful only in partitioned schemes, so in this section we only compare
the three partitioning heuristic against each other and against WF
with load balancing disabled.
The total utilization is varied between 0.5 to 3.5 in steps of
$0.5$. 

\begin{figure}[tb]
  \centering
 \begin{tikzpicture}[scale=0.8]
    \begin{axis}
      [ axis x line=bottom, axis y line = left, xlabel={Total
        Bandwidth}, ylabel={average deadline miss ratio},
	ymin = 0,
	legend entries={"WF-a","BF","FF","WF"},
	legend style={at={(0.5,0.9)}},
yticklabel style={
        /pgf/number format/fixed,
        /pgf/number format/precision=5
},
scaled y ticks=false,
      ]
      \addplot table[y index=2,x index=0]{figs/miss_job_aa.txt};
      \addplot table[y index=1,x index=0]{figs/miss_job_aa.txt};

      \addplot table[y index=4,x index=0]{figs/miss_job_aa.txt};
      \addplot table[y index=3,x index=0]{figs/miss_job_aa.txt};

    \end{axis}
  \end{tikzpicture}
  \caption{Average deadline miss ratio when load balancing is active
    ($w = 20$ jobs).}
  \label{fig:res:miss_aa}
\end{figure}
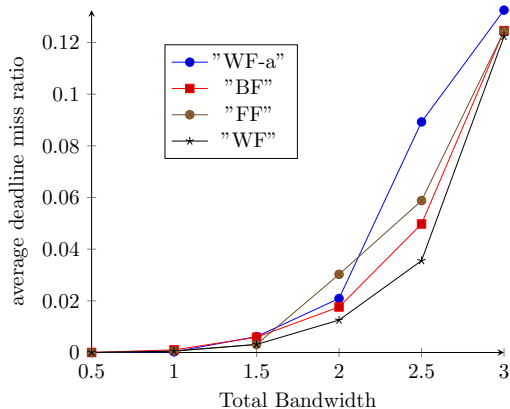

Figure \ref{fig:res:miss_aa} presents the average deadline miss ratio
as a function of total bandwidth. The line labeled WF-a presents the
results when task migration is disabled. In all others, the task
migration is activated, monitoring the deadline miss ratio on time
intervals of size equal to $20$ jobs. As you can see, load balancing
helps to reduce the number of deadline misses for all partitioning
techniques, and once again WF is the best heuristic. When the load is
very high, all cores are heavily loaded and it is more difficult to
apply load balancing, so all heuristics present the same number of
migrations.

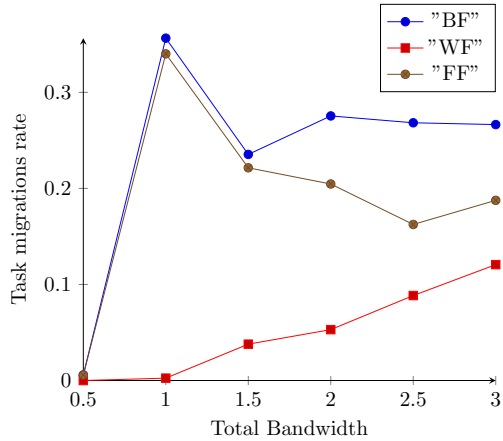
\begin{figure}[tb]
  \centering
  \begin{tikzpicture}[scale=0.8]
    \begin{axis}
      [
	axis x line=bottom,
	axis y line = left,
	xlabel={Total Bandwidth},
	ylabel={Task migrations rate},
	ymin = 0,
	legend entries={"BF","WF","FF"},
	legend style={at={(1,1.1)}},
      ]
      
      \addplot table[y index=3,x index=0]{figs/load.txt};
      \addplot table[y index=1,x index=0]{figs/load.txt};
      \addplot table[y index=2,x index=0]{figs/load.txt};   
    \end{axis}
  \end{tikzpicture}
  \caption{Migrations per job when load balancing is
    active.}\label{fig:load}
\end{figure}

Figure~\ref{fig:load} shows the average number of migrations per job for the same set of
experiments and it confirms the results of the previous figures. In
particular, FF and BF show a bad behavior even at low workloads,
because they tend to concentrate all load on a few processors, and
hence reduce the possibility of local reclaiming, whereas WF ensures
an already good load balancing in the initial allocation.

The number of migrations is greater when total bandwidth is between
1.5 and 2.5.  This is due to the fact that at low utilization, all
deadlines are respected. At high utilizations, all core are highly
loaded and task migrations are not possible, hence all heuristics
behave in the same way. At average load, tasks miss their deadline,
and they have the possibility to move from a core to another one,
hence the number of task migrations number is higher.

\subsection{Dynamic Workloads}
\label{sec:scen-task-insert}

In this subsection we investigate the dynamic behavior of the load
balancing algorithm.  The scenario is generated for a platform with 2
cores (marked with a blue and a pink line, respectively). In the
presented scenario, the task set has a total utilization of $1.2$. The
initial allocation is done using Worst Fit heuristics.  A new task
enters the system at time $t=2000$ and leave it at time $t=6000$. The
new task has a high processor demands. The server utilization for the
new task is set to $0.3$.

\begin{figure}[tb]
  \centering
  \includegraphics[width=0.49\columnwidth]{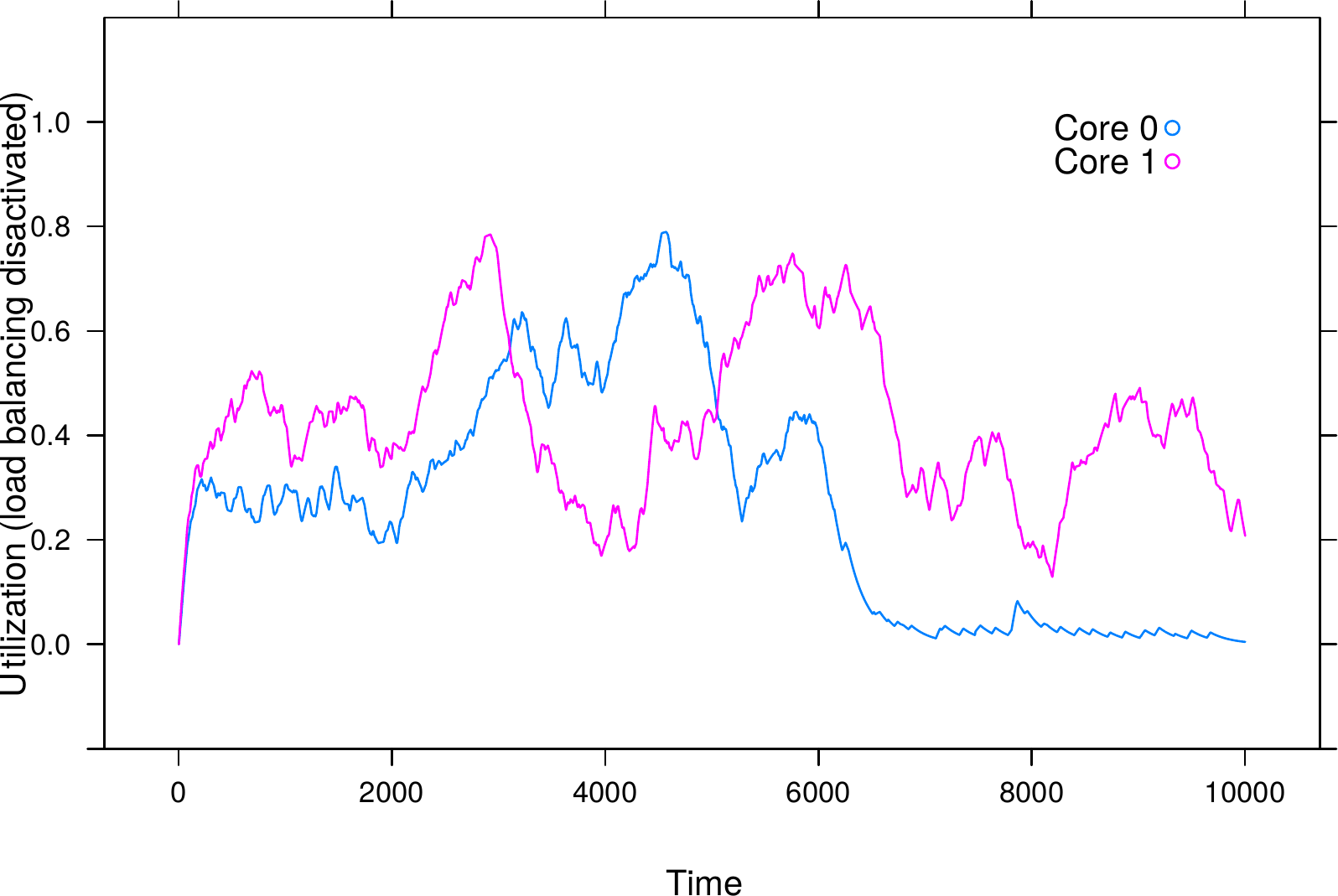}
  \includegraphics[width=0.49\columnwidth]{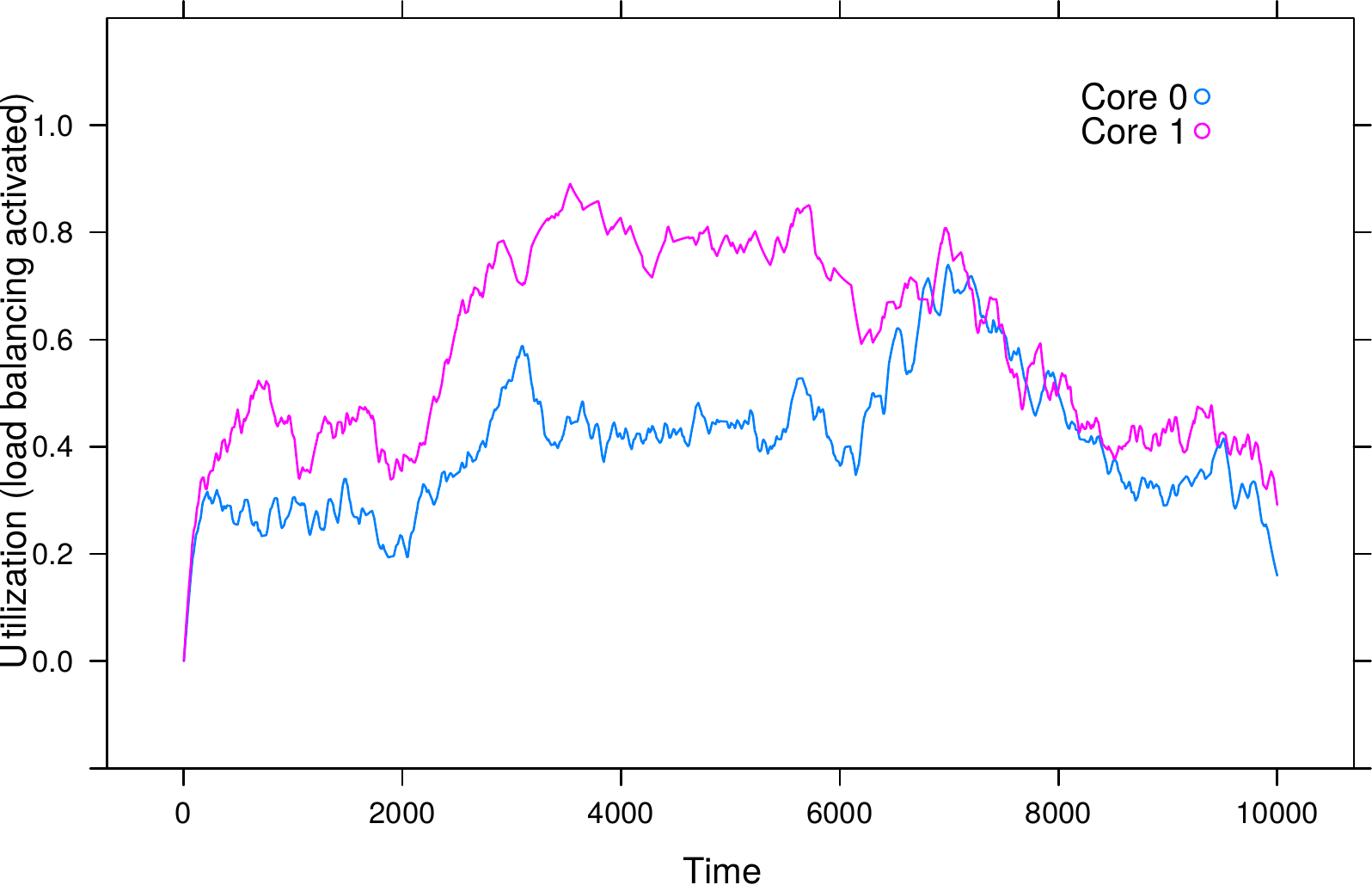}
  \caption{Active utilization when a new task is inserted at $t=2000$, and leaves at $t=6000$.}\label{fig:inserting_load}
\end{figure}

Figure \ref{fig:inserting_load} presents active utilization for every
core, with load balancing enabled (left) and disabled (right). For
clarity of presentation, in the plots we show the exponential average
of the active utilization with ratio $1/200$.

In the case the load balancing is activated, at the beginning tasks
migrate between the two cores thanks to load balancing, and the load
is evenly distributed between cores (left plot). This does not happen
in the right scenario where load balancing is disabled, so the pink
core remains more loaded than the blue one. At $t = 2000$, when the
new task enters the system, in the left scenario the load of the pink
core is increased by an amount equal to the utilization of the new
task, and remains high during the interval $[2000,
6000]$. 
On the other hand, when task migration is activated, cores start
exchanging tasks in order to miss less deadlines. 

\begin{figure}[tb]
  \centering
  \includegraphics[width=0.49\columnwidth]{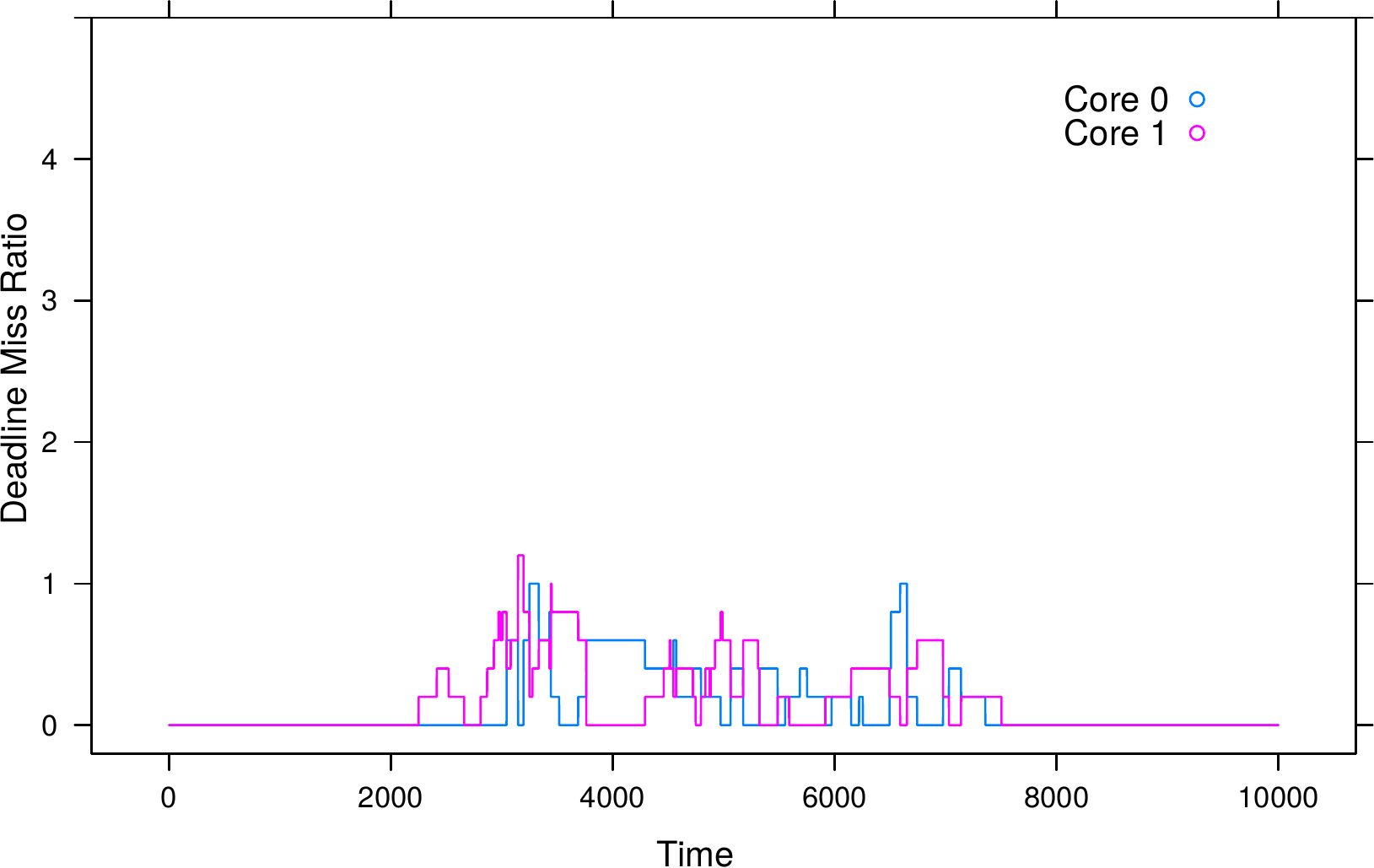}
  \includegraphics[width=0.49\columnwidth]{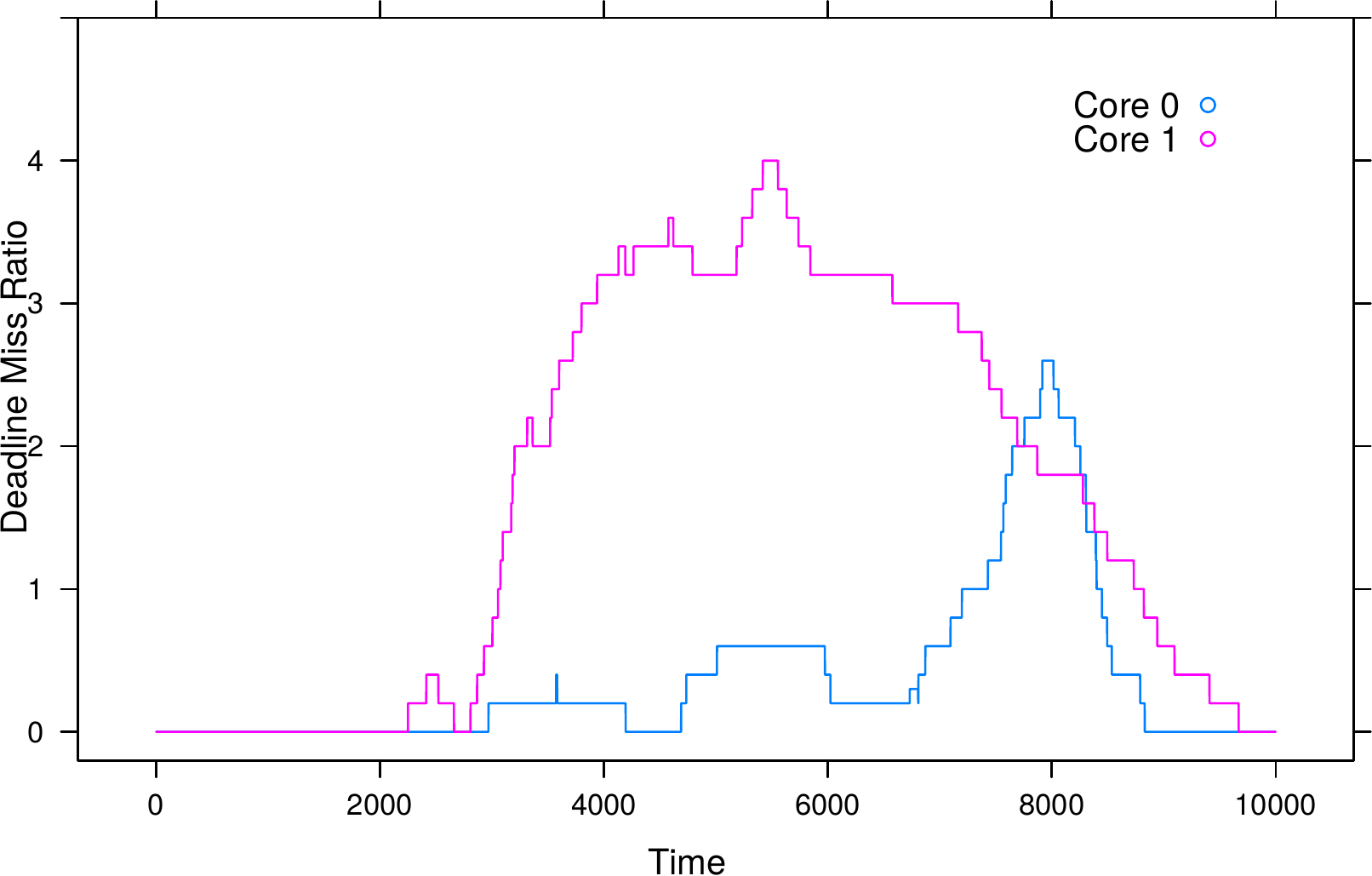}
  \caption{Deadline miss ration on each core for the task inserting scenario}\label{fig:inserting_qos}
\end{figure}

In Figure \ref{fig:inserting_qos}, we present the deadline miss ratio
(fraction of jobs missing their deadline in an interval of time of
size $w$ starting at time $t$) when the load balancing is enabled
(left) and disabled (right) at each instant of time $t$. Before task
inserting, both cores have a deadline miss ratio equals to $0$. When
the task arrives, the deadline miss ratio increases on both
cores. When the task migration is disabled, the pink core has a higher
load and the tasks that are allocated there miss more deadlines
compared to the blue core which has smaller load.  On the other hand,
when load balancing is activated, tasks can permanently migrate from
the blue core to the pink one and vice-versa. Again, the load is
spread more evenly, and this explains the fact that we have a smaller
deadline miss ratio compared to the one where the task migration is
disabled. This does not come for free: the tasks on the blue core miss
slightly more deadlines compared to the case of disabled migration
(figure on the right).  However, as the total bandwidth is balanced in
the activated task migration scenario, the gain on the blue core is
quantitatively higher than the loss on the pink
core. 


\section{Conclusion}
\label{sec:conclusion}

We presented a partitioning scheduling algorithms with temporary
migration for soft real-time tasks. When a task exhausts its budget, it 
first reclaims the unused bandwidth on the local core, and migrates
only if it is necessary to reclaim extra bandwidth on the other cores.
Simulation experiments show that our technique permits to greatly
reduce the number of migrations with respect to global scheduling
without increasing the number of deadline misses. The Worst-Fit
heuristic seems to be the most effective for balancing the load across
cores and to distribute the extra bandwidth. 

As a future work, we plan to implement our algorithm on Linux to
evaluate the overhead of the scheduler and the complexity of the
temporary migration mechanism. We also plan to extend our technique to
heterogeneous multicore processors, such as the ARM big-Little, and
complement it with DVFS and power management schemes, in order to
reduce the energy consumption of modern mobile appliances.


\section{Acknowledgement}
This work has been supported in part by the Institut
de recherche sur les composants logiciels et matriels pour
l information et la communication avancee de Lille, IRCICA,
UMR3380.

\bibliographystyle{unsrt}
\bibliography{mybib}

\end{document}